\documentclass{math-note}

\usepackage[english]{babel}
\usepackage{mathtools,amssymb,amsthm}
\usepackage{tikz}
\usepackage[
  hidelinks,
  pdftitle={Traveling Salesman Problem on Graphs with Large Independent Set},
  pdfauthor={}
]{hyperref}

\title{\huge Traveling Salesman Problem on Graphs with\\Proportional Independent Set}
\author{Ethan X. Fang \quad Shengtao Guo \quad Junwei Lu \quad Yuan Yao\thanks{Corresponding author. 
Department of Biostatistics, Harvard T.H. Chan School of Public Health, Boston, MA 02115, USA. Email: \texttt{yyao1@hsph.harvard.edu}.}}
\date{}

\numberwithin{equation}{section}
\newtheorem{theorem}{Theorem}[section]
\newtheorem{lemma}[theorem]{Lemma}
\newtheorem{proposition}[theorem]{Proposition}
\newtheorem{corollary}[theorem]{Corollary}
\theoremstyle{definition}
\newtheorem{definition}[theorem]{Definition}
\newtheorem{algorithm}{Algorithm}
\theoremstyle{remark}

\newcommand{\R}{\mathbb{R}}
\newcommand{\one}{\mathbf{1}}
\newcommand{\cC}{\mathcal{C}}

\begin{document}
\maketitle

\begin{abstract}
We give a deterministic algorithm for minimum-cost Hamiltonian cycle in
undirected graphs with arbitrary real edge costs and an independent set of
size at least $\alpha n$, for any fixed $3/7<\alpha\le1/2$. It runs in
$O^*(\gamma(\alpha)^n)$ time for an explicit constant $\gamma(\alpha)<2$;
for bipartite graphs, the bound is $O^*((1+2\sqrt2)^{n/2})\approx O^*(1.957^n)$.
The algorithm is
based on a restricted Held--Karp dynamic program. Its state enumeration
uses structural properties of minimum-cost assignments arising from
a regular triangulation of a product of simplices.
\end{abstract}

\section{Introduction}\label{sec:intro}

In the traveling salesman problem (TSP), one is given an undirected graph with real edge
costs and seeks a Hamiltonian cycle, or \emph{tour}, of minimum total
cost. The dynamic program (DP) of Bellman~\cite{Bel62} and of Held and
Karp~\cite{HK62} solves the problem in $O(n^22^n)$ time by storing the
cheapest path from a fixed start for each visited vertex set and endpoint.
Cook~\cite{Coo12} asks the question of whether the exponential base~$2$ can be
reduced for arbitrary costs; see also
the earlier survey of Woeginger~\cite{Woe03} and Nederlof's formulation
\cite[Open Question~1]{Ned20}.

Several faster algorithms exploit restrictions on the input graph or its
edge costs. Stoian~\cite{Sto24} improves the general TSP bound to
$2^nn^2/2^{\Omega(\sqrt{\log n})}$ by batching DP transitions into
min-plus matrix products, without reducing the exponential base.
For graphs of maximum degree three, Eppstein~\cite{Epp07} gives an
$O^*(2^{n/3})$ algorithm that also handles edge costs. More generally,
Bj\"orklund, Husfeldt, Kaski, and Koivisto~\cite{BHKK12} obtain
fixed-base improvements for bounded-degree graphs by trimming DP states.
Rank-based dynamic programming of Bodlaender, Cygan, Kratsch, and
Nederlof~\cite{BCKN15} gives deterministic algorithms for weighted
connectivity problems, including TSP, on graphs of bounded treewidth;
the matching-connectivity framework of Cygan, Kratsch, and
Nederlof~\cite{CKN18} is also used in subsequent work on bipartite TSP.
For Hamiltonicity without edge costs, Bj\"orklund~\cite{Bjo14}
gives an $O^*(1.657^n)$ randomized algorithm, whose extension to integer
costs incurs a factor depending on the cost range.
More recently, Stoian~\cite{Sto26} transfers such unweighted speedups
to TSP when the set of edge costs has bounded doubling, a restriction
not imposed here. For bipartite graphs with arbitrary real costs,
Nederlof~\cite{Ned20} gives a randomized $O(1.9999^n)$ algorithm
assuming nearly quadratic time matrix multiplication.

In this paper, we give a deterministic improvement for bipartite graphs
and, more generally, for graphs with a sufficiently large independent
set. Throughout,
$O^*(\cdot)$ suppresses factors polynomial in~$n$ in the exact arithmetic
model; for binary-encoded rational costs it also suppresses factors
polynomial in the input bit length. For $3/7<\alpha\le1/2$, define
\begin{equation}\label{eq:gamma}
\gamma(\alpha)=\min_{z\ge1}\;2^{2\alpha-1}(2+z+z^2)^{1-\alpha}z^{-\alpha}.
\end{equation}

\begin{theorem}\label{thm:main}
Fix $\alpha\in(3/7,1/2]$. Given an undirected $n$-vertex graph with real
edge costs and containing an independent set of size at least $\alpha n$, a
deterministic algorithm finds a minimum-cost Hamiltonian cycle or reports
that none exists in $O^*(\gamma(\alpha)^n)$ time. It accesses the
input costs only through additions and comparisons, and $\gamma(\alpha)<2$.
\end{theorem}

Here $\alpha$ is fixed independently of~$n$. We describe the algorithm
with a suitable independent set~$R$ supplied; $R$ can be obtained within the same time bound
if needed.

At $\alpha=1/2$, the minimum in~\eqref{eq:gamma} occurs at $z=\sqrt2$,
giving $\gamma(1/2)=\sqrt{1+2\sqrt2}=1.956636\ldots$. A bipartite graph
with a tour has equally sized parts, so we obtain the following consequence.

\begin{corollary}\label{cor:bipartite}
Minimum-cost Hamiltonian cycle in an undirected bipartite graph with
arbitrary real costs can be solved deterministically in
$O^*((1+2\sqrt2)^{n/2})$ time.
\end{corollary}

\subsection{Overview}

Let $L=V\setminus R$, $m=|L|$, and $k=|R|$, and write
$w=(w_{lr})_{l\in L,r\in R}$ for the edge costs between $L$ and $R$. We first
assume that all these edges are present and $w$ is generic; these assumptions are removed in Section~\ref{sec:general}.

\vspace{-10pt}

\paragraph{Assignments in optimal tours.}
An \emph{assignment}
$A\colon R\to L$ chooses one neighbor in~$L$ for every $r\in R$;
its \emph{load} $d$ is the number of vertices of~$R$
assigned to each vertex~$l\in L$. (The terminology comes from standard usage in scheduling problems.) Denote the unique minimum-cost assignment of each
load by~$A_d$.
The structural result (Theorem~\ref{thm:exposure}) is that every
minimum-cost tour contains some~$A_d$ with $d_l\le2$. Its proof uses
auxiliary prices on~$L$, Brouwer's fixed-point theorem, and a
Hamiltonian-cycle parity argument.

\vspace{-10pt}

\paragraph{Structure of minimum-cost assignments.}
The cross-edge cost matrix also defines a tropical hyperplane
arrangement, or equivalently a regular triangulation of a product of simplices
\cite{DS04,GNP24}. We use two known structural results from these studies, phrased in assignment language (Lemma~\ref{lem:neighbor-assignments}): assignments with ``neighboring'' loads are all subgraphs of a
spanning tree $F_d$, and moving one unit of load to a vertex preserves
its previously assigned neighbors~\cite{GNP24,Yao25}.

\vspace{-10pt}

\paragraph{Enumerating compatible subsets.}
A subset~$U\subseteq L\cup R$ is \emph{compatible} with an assignment $A$
if it is a union of the connected components of $A$. Let~$\cC(w)$ consist of the sets
compatible with at least one~$A_d$ of load at most two. Our goal is to generate all sets in $\cC(w)$ efficiently.

For each load~$d$,
we obtain the tree $F_d$ using assignments with neighboring loads.
Contracting edges in $A_d$ and deleting selected edges
gives a forest~$T_d$ on~$L$. Each proper two-coloring of~$T_d$ gives a compatible subset.
The aforementioned structure guarantees that all subsets of $\mathcal{C}(w)$ can be generated this way (Proposition~\ref{prop:enumeration}).

With a standard generating function calculation (Proposition~\ref{prop:count}), the number of compatible subsets is
\[
N_{m,k}
=2^{k-m}\sum_{j=k}^{2m}[z^j](2+z+z^2)^m\le \min_{z\geq 1} \ 2^{k-m}z^{-k}(2+z+z^2)^m \leq \gamma(\alpha)^n,
\]
where $[z^j]P(z)$ is the coefficient of $z^j$ in polynomial~$P(z)$. The corresponding
exponential bound has base~$\gamma(\alpha)<2$ for $\alpha>3/7$.

\vspace{-10pt}

\paragraph{The restricted dynamic program.}
The assignment components in a tour containing~$A_d$ are paths on at
most three vertices, traversed as consecutive blocks. Starting at
$s\in R$ along its assignment edge, every completed-block prefix
belongs to~$\cC(w)$. Therefore a Bellman--Held--Karp dynamic program using
only these vertex subsets, with transitions appending at most three
vertices, retains an optimal tour (Section~\ref{sec:dp}).

\section{The combinatorics of assignments}\label{sec:prelim}

Throughout Sections~\ref{sec:prelim}--\ref{sec:dp}, $G=(L\cup R,E)$ is a simple undirected graph, $R$ is independent, $|L|=m$, $|R|=k$, and $1\le k\le m$. We also assume that all edges of $L\times R$ (i.e.\ \emph{cross-edges}) are present, while the edges inside~$L$ are not restricted. The symbol~$w$ denotes the edge costs and, when used as a matrix, their restriction to~$L\times R$. We write $lr$ for an edge with $l\in L$, $r\in R$, and $\one\in\R^L$ for the all-ones vector.

\subsection{Assignments, prices, and tight graphs}

\begin{definition}
    An \emph{assignment} is a map $A\colon R\to L$, identified with a collection of $k$ cross-edges in $G$. The \emph{load} of this assignment is a nonnegative integer vector $d^A\in \mathbb{N}^L$ where $d^A_l=|A^{-1}(l)|$, and the \emph{cost} $w(A)$ is the sum of the costs of all edges of $A$. A load vector is called \emph{balanced} if $d_l \in \{0, 1, 2\}$ for all $l\in L$.
\end{definition}

In the bipartite case where $m = k$, an assignment is similar to but more flexible than the perfect matchings considered in~\cite{Ned20}.

\begin{definition}
    For a simple cycle $l_1r_1l_2r_2\cdots l_tr_tl_1$ in~$K_{L,R}$, its \emph{alternating sum} is the sum of the costs of alternate edges minus the sum for the remaining edges. Equivalently, it is $\sum_{i=1}^t(w_{l_ir_i}-w_{l_{i+1}r_i})$, where $l_{t+1}=l_1$. The cost matrix~$w$ is \emph{generic} if every such sum is nonzero.
\end{definition} 

\begin{definition}
For a \emph{price vector} $a\in\R^L$, define the \emph{adjusted cost} of a cross-edge by
$\widehat w_{lr}(a)=w_{lr}+a_l$. An edge $lr$ is \emph{tight} (for $a$) if it is one of the edges incident to $r$ with the minimum adjusted cost. The \emph{tight graph} $F(a)$ is the collection of tight edges over all $r\in R$. An assignment is also called \emph{tight} if all of its edges are tight.
\end{definition}
Every $r\in R$ has at least one neighbor in~$L$, and adding a common constant to the prices does not change~$F(a)$. Moreover, the prices change the adjusted cost of all assignments with the same load equally, giving us the following properties:

\begin{lemma}\label{lem:dual}
A tight assignment $A$ has minimum cost among all assignments of load~$d^A$. Equality for another assignment holds only when it is also tight.
\end{lemma}

\begin{lemma}\label{lem:forest}
If $w$ is generic, every tight graph~$F(a)$ is a forest.
\end{lemma}
\begin{proof}
On a cycle of tight edges, the two incident edges at each $r\in R$
have equal adjusted cost. Summing these equalities around the cycle
cancels the prices and gives a zero alternating sum, contrary to
genericity.
\end{proof}

\begin{lemma}\label{lem:canonical}
For generic~$w$ and any nonnegative integer load vector~$d$ with
$\sum_l d_l=k$, the minimum-cost assignment of that load is unique; call the \emph{canonical assignment} $A_d$.
\end{lemma}
\begin{proof}
Reassigning along a cycle alternating between assigned edges and unassigned edges preserves the load, so starting from
a minimum-cost assignment its cost increase is nonnegative, and hence
positive by genericity. The difference of two assignments of the same
load decomposes into alternating cycles; therefore the optimum is unique.
\end{proof}

\subsection{Canonical assignment in optimal tours}

We now show that an optimal Hamiltonian cycle contains a tight assignment.
We use the following consequence of Thomason's ``lollipop''
theorem for Hamiltonian paths~\cite{Tho78}.

\begin{lemma}\label{lem:parity}
Let $K$ be a simple bipartite graph with equally sized parts $X$ and $Y$. If every vertex of~$X$ has degree three, every edge of~$K$ lies on an even number of Hamiltonian cycles. In particular, a Hamiltonian~$K$ has a second Hamiltonian cycle.
\end{lemma}
\begin{proof}
Fix $y_0x_0\in E(K)$ with $x_0\in X$. By
\cite[Theorem~1.1]{Tho78}, the number of Hamiltonian paths beginning
with $y_0,x_0$ and ending at a neighbor of~$y_0$ is even: each
possible endpoint lies in~$X$ and has degree three. Closing the paths
by their last edge to~$y_0$ gives the claimed parity.
\end{proof}

\begin{theorem}\label{thm:exposure}
Let $H$ be a minimum-cost Hamiltonian cycle of~$G$. There exist prices $a\in\R^L$ such that every $r\in R$ is incident to an edge of $H\cap F(a)$. Consequently, $H$ contains an assignment of minimum cost for its load. For generic~$w$, this assignment is a canonical assignment of a balanced load vector.
\end{theorem}
\begin{proof}
For each $r\in R$ and $j\in L$, let
\[
M_r(a)=\min_{i\in N_H(r)}\widehat w_{ir}(a),\ 
\Phi_j(a)=\max_{r\in R}\bigl(M_r(a)-w_{jr}\bigr).
\]
The condition $\Phi(a)\le a$ is equivalent to
$\widehat w_{jr}(a)\ge M_r(a)$ for all $j,r$, and thus to every
$r\in R$ having a tight edge in~$H$.

Fix $s\in L$ and put $B=\max_{i,j\in L,\ r\in R}|w_{ir}-w_{jr}|$. Taking maxima in the definition of~$\Phi$ gives
\[
|\Phi_j(a)-\Phi_s(a)|\le\max_{r\in R}|w_{sr}-w_{jr}|\le B.
\]
Thus the continuous map $\Psi(a)=\Phi(a)-\Phi_s(a)\one$ sends the compact convex box
$\{a:a_s=0,\ |a_j|\le B\text{ for all }j\}$ into itself.
Brouwer's theorem gives a fixed point~$a$, so $\Phi(a)=a+\lambda\one$ with $\lambda=\Phi_s(a)$. It remains to show that $\lambda\le0$.

Suppose $\lambda>0$. For each~$j$, choose $r(j)$ attaining the maximum defining~$\Phi_j(a)$. Then
\[
\widehat w_{j\,r(j)}(a)=M_{r(j)}(a)-\lambda<M_{r(j)}(a),
\]
so $jr(j)\notin H$ and is cheaper in adjusted cost than either tour
edge incident to~$r(j)$.

Subdivide the $m-k$ edges of~$H$ within~$L$ (i.e.\ adding a vertex at the middle of these edges, and putting the $m-k$ added vertices in $R$), and add edge $jr(j)$ for every
$j\in L$. The resulting bipartite graph has two parts of size~$m$,
with degree three at every vertex of~$L$, so Lemma~\ref{lem:parity}
gives another Hamiltonian cycle, and suppressing the degree-two subdivision
vertices yields a different tour~$H'$ retaining all $L$--$L$ edges of~$H$.

At each $r\in R$, the same number of incident edges are added and
removed when we go from $H$ to $H'$. Every added edge has adjusted cost at most $M_r(a)-\lambda$,
while every removed edge has adjusted cost at least $M_r(a)$.
Since $H'$ retains exactly the $L$--$L$ edges of~$H$, each vertex of~$L$ has the same cross-edge degree in both tours. Their price contributions therefore cancel, so
$H'$ has strictly lower total cost than $H$, contradicting optimality.
Thus $\lambda\le0$ and $\Phi(a)\le a$, as desired.
\end{proof}

\subsection{Neighboring assignments and tight trees}\label{subsec:neighbor}

For the remainder of this section and Sections \ref{sec:listing}–\ref{sec:dp}, we assume that the cost matrix $w$ is generic.

Fix an order $L=\{1,\ldots,m\}$. For each (not necessarily balanced) load vector~$d$, let $p = p(d)$ be the smallest index with positive load, and let $L_p =\{p,\ldots,m\}.$
All assignments of load~$d$ are supported on~$L_p$.
For any load vector $d$, one can compute
$A_d$ by minimum-cost perfect matching between~$R$ and $d_l$ copies
of each vertex $l\in L$ in polynomial time~\cite{Kuh55,Mun57}. 
Let $w(d)=w(A_d)$ with mild abuse of notation.

We use the following known properties of minimum-cost assignments. Let $\mathbf e_j \in \mathbb{Z}^L$ denote the unit vector at $j\in L$.

\begin{lemma}\label{lem:neighbor-assignments}
Let $S\subseteq L$ and let $u\in\mathbb Z_{\ge0}^{L}$ be supported
on~$S$, with $\sum_lu_l=k-1$. Then
$\mathcal T_S(u):=\bigcup_{j\in S}A_{u+\mathbf e_j}$
is a spanning tree on~$S\cup R$ and
$\deg_{\mathcal T_S(u)}(l)=u_l+1$ for each $l\in S$.
Given $\mathcal T_S(u)$, each $A_{u+\mathbf e_j}$ can be obtained by assigning each $r\in R$ to the first vertex in $S$ on the unique path from $r$ to $j$ in $\mathcal T_S(u)$.

For every $i,j\in S$ with $i\neq j$, if $r\in R$ is assigned to $i$ in $A_{u+\mathbf e_j}$, then it is also assigned to $i$ in $A_{u+\mathbf e_i}$. (Also, there is exactly one additional vertex assigned to $i$ in $A_{u+\mathbf e_i}$ compared to $A_{u+\mathbf e_j}$.)
\end{lemma}

This lemma follows by applying \cite[Proposition~3.4]{GNP24} on (regular) triangulations of products of two simplices induced by the
restricted cost matrix on $S\times R$; the union formulation is also directly given
in~\cite[Proposition~3.8]{Yao25}. See the next subsection for some further discussion on this connection.

We now apply this lemma with $S=L_p$. The following construction
gives a \emph{normalized price vector} $a^{(d)}$ and a corresponding tight tree $F_d$ for each load $d$.

\begin{lemma}\label{lem:least}
Let $u=d-\mathbf e_p$ and $d^{(j)}=u+\mathbf e_j$ for each $j\in L_p$,
and define
$a^{(d)}_j=w(d)-w(d^{(j)})$.
These are the coordinatewise least prices with $a_p=0$ making $A_d$
tight on $L_p\times R$. The resulting tight graph is
\begin{equation}\label{eq:tree-union}
F_d:=\mathcal T_{L_p}(u)=\bigcup_{j\in L_p}A_{d^{(j)}},
\end{equation}
which is a spanning tree on the vertex set $L_p\cup R$. When $F_d$ is rooted at $p$, every $r\in R$ has parent~$A_d(r)$,
and each $l\in L_p$ has exactly $d_l$ children; each $l>p$ has a
parent in~$R$. 

For fixed~$p$, different load vectors $d$ yield different $a^{(d)}$.
The prices $a^{(d)}$ and tree $F_d$ are computable from $d$ in polynomial time.
\end{lemma}
\begin{proof}
If $A_d(r)=i$, Lemma~\ref{lem:neighbor-assignments} gives $A_{d^{(i)}}(r)=i$. Reassigning $r$
from~$i$ to~$j$ in~$A_{d^{(i)}}$ yields an assignment of load~$d^{(j)}$ whose cost is at least that of~$A_{d^{(j)}}$.
Consequently, for every $j\in L_p$ we have
\[
w(d^{(j)})\le w(d^{(i)})+w_{jr}-w_{ir}
\ \text{ and  }\
\widehat w_{ir}(a^{(d)})\le\widehat w_{jr}(a^{(d)}).
\]
Thus $A_d(r) = i$ is a tight edge for every $r$, and $a^{(d)}_p=0$ by definition (since $d^{(p)} = d$). 

If a different price vector $a'$ with $a'_p=0$ also makes $A_d$
tight, since a tight assignment minimizes total adjusted cost, and $A_{d^{(j)}}$ has load $d-\mathbf e_p+\mathbf e_j$, we have
\[
w(d)+d\cdot a'
\le w(d^{(j)})+(d-\mathbf e_p+\mathbf e_j)\cdot a'.
\]
Thus $a'_j\ge w(d)-w(d^{(j)})=a^{(d)}_j$ for every $j$, showing that $a^{(d)}$ is coordinatewise minimal.

At $a^{(d)}$, each $A_{d^{(j)}}$ has the same adjusted cost as the
tight assignment~$A_d$, so it is tight as well. Their union is the spanning tree $F_d$ by Lemma~\ref{lem:neighbor-assignments}.
By Lemma~\ref{lem:forest}, the tight graph is a forest, so it
equals this tree. Its degree at $p$ is
$u_p + 1 = d_p$ and at each $l>p$ is $u_l + 1 = d_l+1$, giving the desired structure when rooted at $p$.

Given some price $a^{(d)}$, we can compute the tight graph directly, and its vertex degrees allow us to recover $d$. 
Given $d$, one can compute each $A_{d^{(j)}}$ using the perfect matching algorithm, and hence also the tight graph $F_d$ using the union definition in
\eqref{eq:tree-union} in polynomial time.
\end{proof}

\begin{figure}[!ht]
\begin{center}
\begin{tikzpicture}[
  x=1.65cm, y=0.7cm,
  vertex/.style={circle,draw,minimum size=5mm,inner sep=0pt,font=\scriptsize}
]

\newcommand{\drawvertices}[1]{%
  \foreach \i in {1,...,5}
    \node[vertex] (l#1\i) at (0,5-\i) {$l_{\i}$};
  \foreach \j in {1,...,4}
    \node[vertex] (r#1\j) at (1,4.5-\j) {$r_{\j}$};
}

\begin{scope}[shift={(0,6)}]
  \drawvertices{a}
  \draw (la2)--(ra1) (la4)--(ra2) (la4)--(ra3) (la5)--(ra4);
  \node at (0.5,-0.85) {$A_{d^{(2)}}=A_d$};
\end{scope}
\begin{scope}[shift={(2.5,6)}]
  \drawvertices{b}
  \draw (lb3)--(rb1) (lb4)--(rb2) (lb4)--(rb3) (lb5)--(rb4);
  \node at (0.5,-0.85) {$A_{d^{(3)}}$};
\end{scope}
\begin{scope}[shift={(0,0)}]
  \drawvertices{c}
  \draw (lc4)--(rc1) (lc4)--(rc2) (lc4)--(rc3) (lc5)--(rc4);
  \node at (0.5,-0.85) {$A_{d^{(4)}}$};
\end{scope}
\begin{scope}[shift={(2.5,0)}]
  \drawvertices{e}
  \draw (le4)--(re1) (le4)--(re2) (le5)--(re3) (le5)--(re4);
  \node at (0.5,-0.85) {$A_{d^{(5)}}$};
\end{scope}

\begin{scope}[shift={(5,3)}]
  \drawvertices{f}
  \draw (lf2)--(rf1) (lf3)--(rf1) (lf4)--(rf1)
        (lf4)--(rf2) (lf4)--(rf3) (lf5)--(rf3) (lf5)--(rf4);
  \node at (0.5,-1) {$F_d=\displaystyle\bigcup_{j=2}^{5} A_{d^{(j)}}$};
\end{scope}
\end{tikzpicture}
\end{center}
\vspace{-15pt}
\caption{An example of a canonical assignment $A_d$ with load $d = (0, 1, 0, 2, 1)$, its three neighboring canonical assignments, and a tight graph $F_d$ that is the union of these four assignments.}\label{fig:tight-tree}
\vspace{-10pt}
\end{figure}
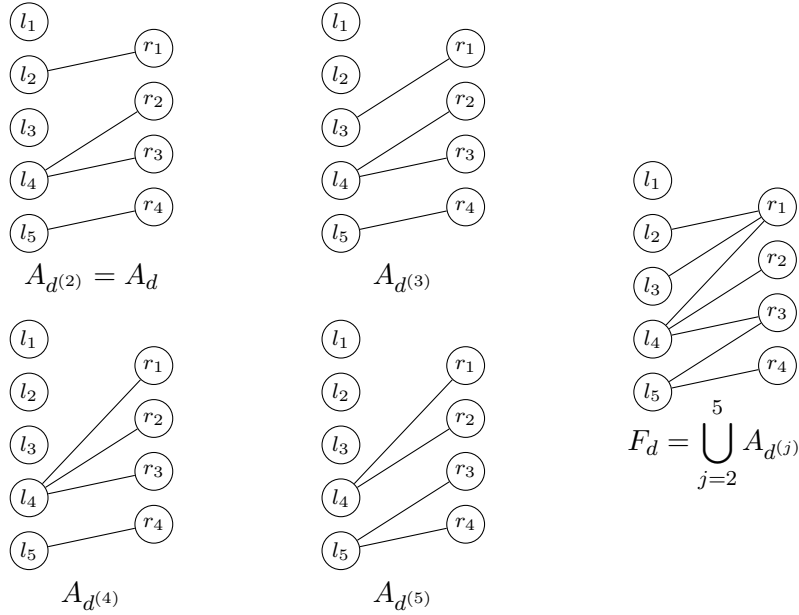

\subsection{Geometric interpretation}

This subsection is intended to provide some geometric background for the assignment construction in this section, and can be skipped safely.

For each $r\in R$, put
$h_r(a)=\min_{l\in L}(w_{lr}+a_l)$.
The locus of $a\in \R^L$ for which this minimum is attained at least twice is called a \emph{tropical hyperplane}. Collecting all tropical hyperplanes for each $r$ gives an \emph{arrangement} of tropical hyperplanes. For each point $a$, the minimizing indices $l$ for each $r$ form the \emph{type} of $a$, representing the relative position of $a$ with respect to the tropical hyperplane arrangement. The tight graph~$F(a)$ encodes the same information as this type using edges of $L\times R$. In particular,
\begin{equation}\label{eq:tropical-product}
\sum_{r\in R}h_r(a)
=\min_{\substack{d\in\mathbb Z_{\ge0}^L\\\sum_l d_l=k}}
\left(w(d)+\sum_{l\in L}d_la_l\right),
\end{equation}
where $w(d)$ is the minimum assignment cost with load~$d$.
The right-hand side is the min-plus tropical polynomial whose
coefficients are the assignment values~$w(d)$. We refer the reader to \cite{DS04} for further connections to tropical geometry.

The same matrix $w$ lifts each vertex $(e_l,e_r)$ of a product of two simplices $\Delta_{m-1}\times\Delta_{k-1}$
to height~$w_{lr}$. Tight graphs encode the lower faces of the
resulting polytope, corresponding to a regular subdivision of $\Delta_{m-1}\times\Delta_{k-1}$, and genericity of $w$ ensures that this is a triangulation. The simplices of this triangulation correspond to bipartite
forests, and its maximal simplices to spanning trees
\cite{GNP24}. 
The combinatorial structure of these triangulations has been studied extensively, and many properties also extend to non-regular triangulations. In this general context, the assignments $A_d$ are
called \emph{semi-matchings} or \emph{topes}, and $\mathcal T_S(u)$ is called a
\emph{tope-linkage covector}, which is always a tree due to the \emph{linkage axiom/property}. We refer the reader to~\cite{Yao25} for a more comprehensive overview.

\section{Enumeration of compatible subsets}\label{sec:listing}

\begin{definition}
A vertex set $U\subseteq L\cup R$ is \emph{compatible} with an assignment~$A$ if
$r\in U$ exactly when $A(r)\in U$, for every~$r\in R$. Equivalently, $U$ is a union of components of the assignment graph, including isolated vertices of~$L$. 

For generic~$w$, let $\cC(w)$ be the family of sets compatible
with some canonical assignment $A_d$ with $d$ balanced. Complementary sets are treated as distinct members of~$\cC(w)$.

For a fixed assignment~$A$, a compatible set can be identified with a $\{0, 1\}$-coloring of $L\cup R$ where every $r\in R$ is given the same color as $A(r)\in L$.
\end{definition}

The key insight to improving the classic Bellman--Held--Karp algorithm is that it suffices to perform dynamic programming over sets in $\cC(w)$, since an optimal Hamiltonian cycle contains some canonical assignment, which decomposes into a collection of connected components, any sub-collection of which belongs to $\cC(w)$. It therefore suffices to generate $\cC(w)$ efficiently.

\subsection{Generating all compatible subsets}

Listing the compatible subsets separately for each assignment would repeat many of them. We describe how to avoid these repetitions below.

Fix a balanced load $d$ and use
the notation $p=p(d)$, $F_d$, and $A_d$ from
Subsection~\ref{subsec:neighbor}.
For each $j>p$ with $d_j\le1$, let $r_j\in R$ be the unique vertex newly
assigned to~$j$ by $A_{d-\mathbf e_p+\mathbf e_j}$ relative to~$A_d$. By 
Lemma~\ref{lem:neighbor-assignments}, $jr_j$ is the unique edge in $F_d$ incident to~$j$ that is not in $A_d$, hence is the parent edge of $j$ in~$F_d$. Define the forest on~$L$
\begin{equation}\label{eq:color-forest}
T_d=\bigl\{\{j,A_d(r_j)\}:j>p,\ d_j\le1\bigr\}.
\end{equation}
Equivalently, $T_d$ is a selectively pruned contraction of the tree $F_d$: contract all assignment edges in~$F_d$, delete the
parent edges at vertices $j>p$ with $d_j=2$, and adjoin
$1,\ldots,p-1$ as isolated vertices. Note that each
component contains exactly one vertex $j\le p$ or with $d_j=2$.

Consider all proper $\{0, 1\}$-colorings of $T_d$ (where adjacent vertices are given opposite colors). For each coloring, if $S\subseteq L$ consists of the
vertices colored $1$, output
$U=S\cup A_d^{-1}(S)$. Observe that $U$ is clearly compatible with $A_d$.

Intuitively speaking, each edge of $T_d$ corresponds to a possible way to reassign a vertex $r\in R$ to get another canonical assignment with a neighboring balanced load. (The deleted edges after contracting $T_d$ by $A_d$ correspond to reassignments that cause a load to be un-balanced.) The proper coloring thus ensures that no valid single-vertex reassignment is compatible with the same coloring. We now use the same intuition to show that every coloring in $\cC(w)$ is still covered despite this restriction.

\begin{figure}[!ht]
\centering
\begin{tikzpicture}[
  vertex/.style={circle,draw,minimum size=6.5mm,inner sep=0pt,font=\small},
  rightvertex/.style={vertex,fill=black!6},
  assignment/.style={line width=1.7pt},
  tight/.style={line width=0.65pt},
  step-arrow/.style={-{Latex[length=2mm]},line width=0.65pt}
]

  \begin{scope}[shift={(0,0)}]
    \node[vertex] (l1a) at (-0.45,4.45) {$l_1$};
    \node[vertex] (l2a) at (1.00,4.45) {$l_2$};
    \node[rightvertex] (r1a) at (1.00,3.60) {$r_1$};
    \node[vertex] (l3a) at (0.00,2.75) {$l_3$};
    \node[vertex] (l4a) at (2.05,2.75) {$l_4$};
    \node[rightvertex] (r2a) at (1.45,1.90) {$r_2$};
    \node[rightvertex] (r3a) at (2.65,1.90) {$r_3$};
    \node[vertex] (l5a) at (2.65,1.05) {$l_5$};
    \node[rightvertex] (r4a) at (2.65,0.20) {$r_4$};
    \draw[tight] (r1a)--(l3a) (r1a)--(l4a) (r3a)--(l5a);
    \draw[assignment] (l2a)--(r1a) (l4a)--(r2a)
                      (l4a)--(r3a) (l5a)--(r4a);
    \node[font=\small] at (1.10,-0.5) {$F_d$};
  \end{scope}

  \draw[->] (3.02,2.50)--(4.35,2.50)
    node[midway,above,font=\scriptsize] {contract $A_d$};

  \begin{scope}[shift={(5.00,0)}]
    \node[vertex] (l1b) at (-0.45,4.45) {$l_1$};
    \node[vertex] (l2b) at (1.00,4.45) {$l_2$};
    \node[vertex] (l3b) at (0.00,2.75) {$l_3$};
    \node[vertex] (l4b) at (2.05,2.75) {$l_4$};
    \node[vertex] (l5b) at (2.65,1.05) {$l_5$};
    \draw[tight] (l2b)--(l3b) (l2b)--(l4b) (l4b)--(l5b);
  \end{scope}

  \draw[->] (8.02,2.50)--(9.35,2.50)
    node[midway,above,align=center,font=\scriptsize] {delete parent\\edge of $l_4$};

  \begin{scope}[shift={(10.00,0)}]
    \node[vertex] (l1c) at (-0.45,4.45) {$l_1$};
    \node[vertex] (l2c) at (1.00,4.45) {$l_2$};
    \node[vertex] (l3c) at (0.00,2.75) {$l_3$};
    \node[vertex] (l4c) at (2.05,2.75) {$l_4$};
    \node[vertex] (l5c) at (2.65,1.05) {$l_5$};
    \draw[tight] (l2c)--(l3c) (l4c)--(l5c);
    \node[font=\small] at (1.10,-0.5) {$T_d$};
  \end{scope}
\end{tikzpicture}
\caption{Constructing the forest $T_d$ from the example tight graph $F_d$ from Figure~\ref{fig:tight-tree}. The contracted edges of $A_d$ are bolded in $F_d$.}\label{fig:construct-forest}
\end{figure}
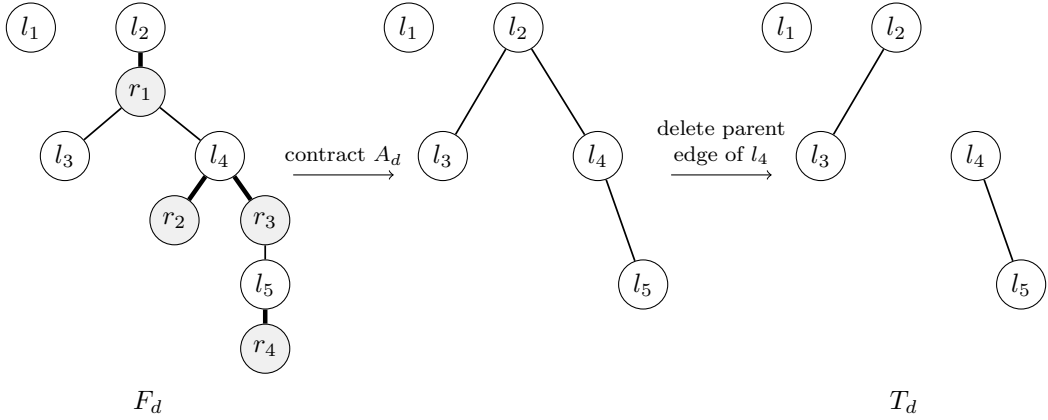

\begin{proposition}\label{prop:enumeration}
Over all balanced load vectors $d$, the process above generates every set in~$\cC(w)$ at least once. Its number of outputs, counted with multiplicity, is $N_{m,k}:=\sum_d 2^{b(d)}$, where $b(d):=p(d)+|\{l>p(d):d_l=2\}|$ is the number of connected components of $T_d$.
The list can be generated deterministically in $O^*(N_{m,k})$ time.
\end{proposition}

\begin{proof}

Fix $U\in\cC(w)$ and choose a compatible balanced load~$d$ maximizing $p=p(d)$ and then minimizing $\sum_{j\in L_p}a^{(d)}_j$.
Assume for the sake of contradiction that an edge $lA_d(r_l)$ of~$T_d$ has both endpoints on the same side of~$U$ (where $l$ is further away from the root in $F_d$), then clearly $r_l$ is also on the same side as well. 
Now reassign $r_l$ from $A_d(r_l)$ to~$l$.
The old and new assignment edges both lie in~$F_d$, so the new assignment~$A'$ is compatible with~$U$ and minimum-cost for its load~$d'\ne d$: indeed, every assignment of load~$d'$ is supported on~$L_p$. Thus $A'=A_{d'}$ by Lemmas~\ref{lem:dual} and~\ref{lem:canonical}, and its loads remain at most two.
As $p(d')\ge p$, the choice of~$d$ forces $p(d')=p$. Since $A_{d'}$ is tight for~$a^{(d)}$, Lemma~\ref{lem:least} gives $a^{(d')}\le a^{(d)}$ coordinatewise, with equality only if $d'=d$. Hence the sum of prices strictly decreases, which is a contradiction. Therefore, $U$ induces a proper coloring on $T_d$, so the process generates~$U$.

Now list the balanced loads and store their
canonical assignments, with polynomial work per load. These give each forest~$T_d$ in polynomial time
and its $2^{b(d)}$ colorings with polynomial work per output. Since every
load contributes at least one output, total time is $N_{m,k}\,n^{O(1)}$.
\end{proof}

\subsection{Non-repetition}

In fact, with some more analysis we can show that every set is generated \emph{exactly} once using the process in the previous section, making it unnecessary to deduplicate the outputs. This subsection does not affect the algorithm or the runtime analysis, and can be skipped safely.

\begin{proposition}\label{prop:nonrepetition}
The same process as Proposition~\ref{prop:enumeration} generates every set in~$\cC(w)$ without repetition. Therefore, $|\cC(w)|=N_{m,k}$.
\end{proposition}
\begin{proof}
Suppose $U$ is generated by distinct balanced loads $d,d'$, with $p=p(d)\le p'=p(d')$. If $p=p'$, the prices are distinct by Lemma~\ref{lem:least}; interchange the loads if necessary so that $\max_j(a^{(d)}_j-a^{(d')}_j)>0$.
If $p<p'$, subtract a sufficiently large common constant from~$a^{(d')}$; this preserves its tight graph and the argument below.

Let $\delta=\max_{j\in L_{p'}}(a^{(d)}_j-a^{(d')}_j)>0$, and choose a maximizing vertex $l\in L_{p'}$ closest to~$p'$ in~$F_{d'}$. Note that $l>p$; the case $l=p=p'$ is impossible since $a_{l}^{(d)} = a_{l}^{(d')} = 0$ by definition of normalized prices.

Every edge~$lr$ of~$F_d$ incident to~$l$ is also tight for~$a^{(d')}$: for $i\in L_{p'}$,
\[
a^{(d')}_l+w_{lr}=a^{(d)}_l+w_{lr}-\delta
\le a^{(d)}_i+w_{ir}-\delta\le a^{(d')}_i+w_{ir}.
\]
If $d_l=2$, then $l$ has three tight incident edges in $F_d$ and hence also in $F_{d'}$,
so $l\ne p'$ (since $p'$ only has degree 2) and there are no other incident edges in $F_{d'}$ (since the maximum degree is 3), 
and therefore its parent edge in $F_{d'}$ is common to both trees.
Otherwise, its parent edge in $F_d$ crosses~$U$ and cannot belong to the
compatible assignment~$A_{d'}$. This edge must therefore also be the parent edge in $F_{d'}$, again common to both trees and with $l\ne p'$.

In both cases, let $lr$ be this common parent edge and $j=A_{d'}(r)$ (i.e.\ the parent of $r$ in $F_{d'}$). Since $lr$ is tight for~$a^{(d)}$, while $lr$ and $jr$ are tight for~$a^{(d')}$, we have
\[
w_{lr}+a^{(d)}_l\le w_{jr}+a^{(d)}_j,\ 
w_{lr}+a^{(d')}_l=w_{jr}+a^{(d')}_j.
\]
Subtracting the two gives $\delta=a^{(d)}_l-a^{(d')}_l\le a^{(d)}_j-a^{(d')}_j\le\delta$,
so $j$ is also a maximizing vertex, but $j$ is closer to~$p'$ than~$l$ in $F_{d'}$, a contradiction to the choice of $l$. Hence different loads have disjoint outputs. It is also clear no load generates the same set more than once.
\end{proof}

\subsection{Upper bound of the number of outputs}\label{sec:count}

\begin{proposition}\label{prop:count}
For $1\le k\le m$, the number of outputs in Proposition~\ref{prop:enumeration} is
\begin{equation}\label{eq:count}
N_{m,k}=2^{k-m}\sum_{j=k}^{2m}[z^j](2+z+z^2)^m.
\end{equation}
Consequently, for every $z\ge1$,
\begin{equation}\label{eq:count-bound}
|\cC(w)|= N_{m,k}\le 2^{k-m}z^{-k}(2+z+z^2)^m.
\end{equation}
\end{proposition}

\begin{proof}
Introduce a formal variable $z$ to record the total load, with each choice of $d_l$ contributing a factor $z^{d_l}$. Let $P(z) = 1+z+2z^2$. Grouping the loads by their first positive index~$p$ and considering the contributions from vertices $l < p$, $l = p$, and $l > p$ separately, it is not difficult to see that 
\begin{align*}
N_{m,k} = \sum_{d} 2^{b(d)}
&=[z^k]\sum_{p=1}^m2^{p-1}(2z+2z^2)P(z)^{m-p}\\
&=[z^k]\,2z(1+z)\frac{2^m-P(z)^m}{2-P(z)}=[z^k]\frac{2z(2^m-P(z)^m)}{1-2z},
\end{align*}
where we use $2-P(z)=(1+z)(1-2z)$ for the last equality.
Since $P(1/2)=2$, we have
$2^m=\sum_{j=0}^{2m}2^{-j}[z^j]P(z)^m$. Expanding
$2z(1-2z)^{-1}=\sum_{t\ge1}(2z)^t$ therefore gives
\begin{align*}
N_{m,k}
= [z^k] (2^m - P(z)^m)\sum_{t\geq 1} (2z)^t 
&= 2^{k+m} - \sum_{t\geq 1}[z^{k-t}] 2^{t}P(z)^m \\
&= 2^k\Bigl(2^m-\sum_{j=0}^{k-1}2^{-j}[z^j]P(z)^m\Bigr)  =2^{k} \sum_{j=k}^{2m}2^{-j}[z^j]P(z)^m.
\end{align*}

Finally, $2P(z/2)=2+z+z^2$, so
$[z^j](2+z+z^2)^m=2^{m-j}[z^j]P(z)^m$, giving~\eqref{eq:count}.
The bound~\eqref{eq:count-bound} follows by using $z^j\ge z^k$
for $j\ge k$ and real~$z\ge1$.
\end{proof}

\paragraph{Exponential bound.}\label{subsec:evaluation}
For $k/n=\alpha$, the right side of~\eqref{eq:count-bound} is
\begin{equation*}
\left[2^{2\alpha-1}(2+z+z^2)^{1-\alpha}z^{-\alpha}\right]^n.
\end{equation*}
At $z=1$ the bracket equals~$2$, and its logarithmic derivative with respect to~$z$ at $z = 1$ is $(3-7\alpha)/4$. This is negative for $\alpha>3/7$, so $\gamma(\alpha)<2$. The expression tends to infinity as $z\to\infty$ for $\alpha\le1/2$, so the minimum is attained.

For fixed $z\ge1$, increasing~$\alpha$ multiplies the expression in~\eqref{eq:gamma} by a positive power of $4/[z(2+z+z^2)]\le1$, so the function~$\gamma$ is nonincreasing for~$\alpha \in (3/7,1/2]$. In particular, if $k/n\ge\alpha$, then $N_{m,k}\le\gamma(\alpha)^n$.

\section{Restricted dynamic programming}\label{sec:dp}

We restrict the Held–Karp state space to compatible subsets $\cC(w)$. Since every connected component of an assignment has at most three vertices, transitions may append up to three vertices at a time, and can be computed in polynomial time.

We assume $n\ge3$ and fix an arbitrary start vertex $s\in R$.
Use states $(U,v)$ with $U\in\cC(w)$, $s\in U$, and $v\in U\setminus\{s\}$,
together with the initial state $(\{s\},s)$. Each entry stores the minimum
cost of an admitted path from~$s$ to~$v$ visiting exactly the set~$U$.
Set $\mathrm{DP}(\{s\},s)=0$ and all other entries to~$+\infty$.

Process states in the order of increasing $|U|$. From $\mathrm{DP}(U,u)$, consider every simple path $P=(u,x_1,\ldots,x_t)$ with $1\le t\le3$ and $x_i\notin U$. Let $U'=U\cup\{x_1,\ldots,x_t\}$; if $(U',x_t)$ is a stored state, update
\begin{equation}\label{eq:dp}
\mathrm{DP}(U',x_t)\leftarrow
\min\bigl\{\mathrm{DP}(U',x_t),\ \mathrm{DP}(U,u)+w(P)\bigr\}.
\end{equation}
All paths use available edges of~$G$. Close each full path by an edge
back to~$s$: the answer is
\[
\min_{v\in N_G(s)}\bigl(\mathrm{DP}(V,v)+w_{vs}\bigr),
\]
and the tour can be recovered by looking at predecessors. An infinite value means no tour exists.

\begin{proposition}\label{prop:generic}
Given generic costs~$w$ and an independent set~$R$, the procedure solves TSP deterministically in $O^*(N_{m,k})$ time.
\end{proposition}
\begin{proof}
Every finite entry represents a simple path with vertex set~$U$,
so closing a full path gives a feasible tour.

By Theorem~\ref{thm:exposure}, a minimum-cost tour~$H$ contains some~$A_d$.
Orient~$H$ to traverse the assignment edge at~$s\in R$ first. The
connected components of~$A_d$ are consecutive blocks of at most three vertices,
and~$s$ is an endpoint of its block. Complete this first block from the
initial state, then append the remaining blocks. Every block union is
compatible with~$A_d$, so~\eqref{eq:dp} admits this path and its closing edge.

Generate the sets in $\cC(w)$ using
Proposition~\ref{prop:enumeration}. Store the subsets by cardinality in a
deterministic dictionary keyed by $n$-bit strings. There are
$O(nN_{m,k})$ states, each with $O(n^3)$ candidate extensions,
giving $N_{m,k}n^{O(1)}$ time.
\end{proof}

The algorithm below includes the step for handling arbitrary inputs, justified in Section~\ref{sec:general}. On such inputs the family is $\cC(\widetilde w)$ for the completed symbolic costs, while the dynamic program uses original edges and costs.

\begin{algorithm}[Deterministic TSP with an independent set]\label{alg:tsp}
Given a simple graph $G=(V,E)$ with costs~$w$ and an independent set
$R$ of size at least $\alpha|V|$, proceed as follows.
\begin{enumerate}
\item Set $L=V\setminus R$, $m=|L|$, and $k=|R|$; if $k>m$ or $|V| < 3$, report no tour and stop.
\item Fix $s\in R$, label $L=\{1,\ldots,m\}$, and construct the
completed costs $\widetilde{w}$ as in~\eqref{eq:lex-weight}. Initialize an empty deterministic
dictionary $\cC(\widetilde{w})$ of vertex subsets.
\item List all loads $d\in\{0,1,2\}^L$ with $\sum_l d_l=k$.
Compute and store each~$A_d$ by minimum-cost perfect matching with
costs~$\widetilde w$.
\item For each listed load, form the forest~$T_d$ from~\eqref{eq:color-forest}.
For every proper two-coloring $c\colon L\to\{0,1\}$ of~$T_d$ with
$c(A_d(s))=1$, put
$S=c^{-1}(1)$ and insert $S\cup A_d^{-1}(S)$ into the dictionary.
\item Run~\eqref{eq:dp} with $\mathrm{DP}(\{s\},s)=0$, using
the original edges and costs. Minimize $\mathrm{DP}(V,v)+w_{vs}$ over
edges $vs\in E$ and recover a tour, or report nonexistence if the
minimum is infinite.
\end{enumerate}
\end{algorithm}

If needed, obtain~$R$ by computing a maximum independent set in
$O^*(2^{n/3})$ time~\cite{TT77}. Its time is absorbed by
the stated bound since $2^{1/3}<\gamma(1/2)\le\gamma(\alpha)$.

\section{Arbitrary costs and edges}\label{sec:general}

We remove the assumptions of complete cross-edges and generic costs by introducing deterministic perturbations to the costs.

Fix an independent set $R$ with $k=|R|\le m=|L|$.
Add all missing $L$--$R$ edges, retaining the original edges within~$L$
and adding none within~$R$. Assign distinct nonnegative identifiers
$\iota(e)<n^2$ to the completed graph edges and use lexicographically
ordered costs
\begin{equation}\label{eq:lex-weight}
\widetilde w_e=
\begin{cases}
(0,w_e,2^{\iota(e)}),&e\in E,\\
(1,0,2^{\iota(e)}),&e\notin E.
\end{cases}
\end{equation}
The first coordinate counts added edges, the second gives the original
cost, and the third breaks ties. Thus, whenever an original tour exists,
every symbolic minimum-cost tour is an original optimum. Every alternating
cycle has a nonzero third coordinate, since distinct powers of two have
nonzero signed sum, so the cross-edge matrix is generic. 

This is a symbolic perturbation in the spirit of simulation of
simplicity~\cite{EM90}. The next lemma supplies the real realization
needed for the structural arguments.

\begin{lemma}\label{lem:realization}
For each fixed input, real edge costs with a generic cross-edge matrix can reproduce the canonical assignments, coloring forests, and minimum-cost tours determined by~\eqref{eq:lex-weight}. The algorithm does not compute such a realization.
\end{lemma}
\begin{proof}
Replace each triple $(x_0,x_1,x_2)$ by
$x_0+\varepsilon x_1+\varepsilon^2x_2$. For sufficiently small
$\varepsilon>0$, this preserves all comparisons among the finitely
many assignment and tour costs. It also preserves genericity: every
alternating-cycle sum has a nonzero third coordinate. Thus the
canonical assignments and optimal tours are unchanged, and so are
the coloring forests defined by~\eqref{eq:color-forest}.
\end{proof}

The algorithm works directly with triples, using coordinatewise addition
and subtraction and lexicographic comparison. The integer coordinates
have polynomial bit length; only the middle coordinate involves input
reals. To implement subtraction using additions and comparisons, store
each signed expression as a difference of sums $p-q$: addition is
componentwise, negation swaps the sums, and $p-q\le p'-q'$ iff
$p+q'\le p'+q$. Each matching computation uses polynomially many such
operations, and the dynamic program uses only sums of original costs.
For rational inputs all these operations have polynomial bit complexity.

\begin{proof}[Proof of Theorem~\ref{thm:main}]
Run Algorithm~\ref{alg:tsp} with $|R|\ge\alpha n$. If $k>m$, no tour
exists; otherwise $\alpha n\le k\le n/2$. Whenever an original tour
exists, a symbolic optimum is an original optimum. Lemma~\ref{lem:realization}
and the block argument of Proposition~\ref{prop:generic} show that the
dynamic program admits such a tour. Every returned tour uses original
edges, so the answer is correct.

The time is $O^*(N_{m,k})\le O^*(\gamma(\alpha)^n)$ by
Proposition~\ref{prop:count} and monotonicity. The optional preprocessing and
arithmetic discussion above give the remaining guarantees.
\end{proof}

\paragraph{AI use.} The Odin Automatic AI Research Agent was used to find the initial proof. The final content and proofs in this paper were written and verified by the authors.


\end{document}